\documentclass{sig-alternate}

\mathcode`l="8000
\begingroup
\makeatletter
\lccode`\~=`\l
\DeclareMathSymbol{\lsb@l}{\mathalpha}{letters}{`l}
\lowercase{\gdef~{\ifnum\the\mathgroup=\m@ne \ell \else \lsb@l \fi}}%
\endgroup

\usepackage[utf8]{inputenc}   %
\usepackage[T1]{fontenc}      %
\usepackage{lmodern}

\usepackage[shortlabels]{enumitem}
\setlist[description]{font=\normalfont\itshape,itemsep=0ex,partopsep=0ex}

\usepackage{microtype}  %
\usepackage{mathtools}  %
\usepackage{booktabs}   %
\usepackage{array}      %
\usepackage[nocompress,nosort]{cite}

\usepackage{float}
\usepackage[noend]{algpseudocode}
\newfloat{algo}{tb}{lop}
\floatname{algo}{Algorithm}
\usepackage{caption}  %
\usepackage{comment}  %
\usepackage{amssymb,amsmath}
\usepackage{amsfonts}
\usepackage{xcolor}
\usepackage{url}
\usepackage[colorlinks=true,citecolor=blue,hypertexnames=false]{hyperref}

\newtheorem{thm}{Theorem}
\newtheorem{prop}[thm]{Proposition}
\newtheorem{lem}[thm]{Lemma}
\newtheorem{coro}[thm]{Corollary}

\newtheorem{remark}[thm]{Remark}

\newenvironment{algoenv}[3][\linewidth]{
\begin{minipage}{#1}%
\flushleft
\rule{\textwidth}{.08em}\vspace{-\baselineskip}\smallskip
\begin{description}[noitemsep]
\item[\rlap{Input}\phantom{Output}] #2
\item[Output] #3
\end{description}
\vspace{-\baselineskip}
\rule{\textwidth}{.05em}
\begin{algorithmic}
}{\end{algorithmic}
\vspace{-.5\baselineskip}
\rule{\textwidth}{.08em}
\end{minipage}}

\newcommand{\Reduce}{\mathop{\textsc{Reduce}}}
\newcommand{\tReduce}{\textsc{Reduce}}
\newcommand{\tCreatTel}{\textsc{Telesc}}
\newcommand{\CreatTel}{\mathop{\textsc{\tCreatTel}}}
\newcommand{\mop}[1]{\operatorname{#1}}
\newcommand{\Lx}{L[\mathbf x]}
\newcommand{\Lwx}{L(\varepsilon)[\mathbf x]}
\DeclareMathOperator{\Jac}{Jac}
\newcommand{\ud}{\mathrm{d}}

\DeclareMathOperator{\id}{id}

\newcommand{\softO}{\tilde{\mathcal{O}}}
\newcommand{\bigO}{{\mathcal{O}}}
\newcommand{\PP}{\mathbb{P}}
\DeclareMathOperator{\Mat}{Mat}
\newcommand{\xx}{\mathbf{x}}
\newcommand{\QQ}{\mathbb{Q}}
\newcommand{\CC}{\mathbb{C}}
\newcommand{\pr}{{\smash{\mathrm{pr}}}}
\newcommand{\aff}{{\smash{\mathrm{aff}}}}
\newcommand{\reg}{{\smash{\mathrm{reg}}}}
\newcommand{\Hfp}{H_f^{\pr}}
\newcommand{\Lxfp}{\Lxf_{\smash{-n-1}}}
\newcommand{\Lxf}{L[\xx,\smash{\tfrac{1}{f}}]}
\def\le{\leqslant}

\newcommand{\eexp}{\mathrm{e}}

\newcommand{\GD}{Griffiths--Dwork}

\title{Creative Telescoping for Rational Functions \\ Using the Griffiths--Dwork Method\titlenote{This work has been supported in part by the Microsoft Research\,--\,Inria Joint Centre.}}

\numberofauthors{3}
\author{
  \alignauthor
  Alin Bostan\\
  \affaddr{Inria (France)}\\
  \email{alin.bostan@inria.fr}
  \alignauthor
  Pierre Lairez\\
  \affaddr{Inria (France)}\\
  \email{pierre.lairez@inria.fr}
  \alignauthor
  Bruno Salvy\\
  \affaddr{Inria (France)}\\
  \email{bruno.salvy@inria.fr}
}

\begin{document}
\clubpenalty=10000 
\widowpenalty = 10000

\setlength{\belowdisplayskip}{.3\baselineskip} \setlength{\belowdisplayshortskip}{0pt}
\setlength{\abovedisplayskip}{.36\baselineskip} \setlength{\abovedisplayshortskip}{0pt}

\maketitle

\begin{abstract}
Creative telescoping algorithms compute linear differential equations satisfied by multiple integrals with parameters.
We describe a precise and elementary algorithmic version of the Griffiths--Dwork method for the creative telescoping of rational functions. 
This leads to bounds on the order and degree of the coefficients of the differential equation, and to the first complexity result which is single exponential in the number of variables.
One of the important features of the algorithm is that it does not need to compute certificates. The approach is  vindicated by a prototype implementation.
\end{abstract}

\vspace{1mm}
\noindent
{\bf Categories and Subject Descriptors:} \\
\noindent I.1.2 [{\bf Computing Methodologies}]: Symbolic and
Algebraic Manipulations --- \emph{Algebraic Algorithms}

\vspace{1mm}
\noindent {\bf General Terms:} Algorithms, Theory.

\vspace{1mm}
\noindent {\bf Keywords:} Integration, creative telescoping, algorithms, complexity, Picard-Fuchs equation, Griffiths--Dwork method

\section{Introduction}

\noindent In computer algebra, \emph{creative telescoping} is an approach introduced by Zeilberger to address definite summation and integration of a large class of functions and sequences~\cite{Zei90,Zei91,WilZei92}. %
Its vast scope includes the computation of differential equations for multiple integrals  of rational or algebraic functions with parameters.
Within this class, creative telescoping is similar to well-studied older approaches whose key notion is the Picard--Fuchs differential equation, see \emph{e.g.}~\cite{Pic99}.

We study the multivariate rational case: Given a rational function~$F(t,x_1,\dots,x_n)$, we aim at finding~$n$  rational functions~$A_i(t,x_1,\dots,x_n)$ and a differential operator~$T$ with polynomial coefficients, say~$\sum_{j=0}^r c_j(t) \partial_t^j$, such that
\begin{equation}\label{eqn:prob-ct}
T(F)\stackrel{\text{def}}{=}\sum_{j=0}^{\smash{r}} c_j(t) \partial_t^jF = \sum_{i=1}^{\smash{n}} \partial_i A_i,
\end{equation}
where~$\partial_t^j$ denotes~$\smash{\frac{\partial^j}{\partial t^j}}$ and~$\partial_i$ denotes~$\smash{\frac{\partial}{\partial x_i}}$.
The operator~$T$ is a \emph{telescoper} and the tuple~$(A_1,\dotsc,A_n)$ is a \emph{certificate} for~$T$.
The integer~$r$ is the \emph{order} of~$T$ and $\max_j \deg c_j$ is its \emph{degree}. 

Throughout the article, the constant field $k$ of $F$ is assumed to be of characteristic zero.
Under suitable additional hypotheses, $T(I)=0$ is a differential equation satisfied by an integral~$I(t) = \int F \ud x$ over a domain~$\gamma$, without boundaries, where~$F$ has no pole.
A misbehavior may occur when the certificate has poles outside those of~$F$:
it may not be possible to integrate term by term the right-hand side of Equation~\eqref{eqn:prob-ct},
see~\S\ref{sec:regular-certificate}.
The certificate is called \emph{regular} when it does not contain poles other than those of~$F$.
For integration, there is no need to compute the certificate provided that it is regular. 

Several methods are known that can find a telescoper and the corresponding certificate~\cite{Lip88,Takayama1990a,Chy00,Kou10}.
However, the practical cost of using these methods in multivariate problems remains high and a better understanding of the size or complexity of the objects of creative telescoping is clearly needed.
The present work is part of the on-going
effort in this direction~\cite{ApaZei06,BosCheChy10,CheKauSin12}.
The study of the rational case is motivated both by its fundamental nature and by its applications to the computation of diagonals in combinatorics, number theory and physics~\cite{Lip88,Chr85,Mor92}.
The rational case with~$n$ variables also includes the algebraic case with~$n-1$ variables~\cite{CheKauSin12}.

\medskip
\noindent\emph{Previous works.}
An obviously related problem is, given a rational function~$F(x_1,\dotsc,x_n)$, to decide whether there exist rational functions $A_1,\dotsc,A_n$ such that~$F$ equals~$\smash{\sum_{i=1}^n \partial_i A_i}$.

When~$n=1$, this question is easily solved by Hermite reduction.
This is the basis of an algorithm for creative telescoping~\cite{BosCheChy10} that we outline in~\S\ref{sec:hermite}.
Picard~\cite[chap.~7]{PicSim} gave methods when~$n=2$ from which he deduced that a telescoping equation exists in that case~\cite{Pic02}. This too has led to an algorithm~\cite{CheKauSin12}. 
The Griffiths--Dwork method~[\citen{Dwo62}, \S 3; \citen{Dwo64}, \S 8; \citen{Gri69}] solves the problem for a general~$n$, in the setting of de~Rham cohomology and under a regularity assumption.
The method can be viewed as a generalization of Hermite reduction.
Independently, Christol used a similar method to prove that diagonals of rational functions, under a regularity hypothesis, are differentially finite~\cite{Chr83}; then he applied a deformation technique, for which he credits Dwork, to handle singular cases~\cite{Chr85}.
The Griffiths--Dwork method is also used in point counting~\cite{AbbKedRoe10,Gerkmann07} and the study of mirror maps~\cite{Mor92}.

In terms of complexity, in more than two variables, not 
much is known.
If a rational function~$F(t,x_1,\dotsc,x_n)$ has degree~$d$, a study of Lipshitz's argument~\cite{Lip88} shows that there exists a telescoper of order and degree~$d^{O(n)}$ with a regular certificate of size~$d^{\smash{O(n^2)}}$.
Most algorithms~\cite{Lip88,Zei90,Takayama1990a,Chy00,ApaZei06,Kou10} cannot avoid the computation of the certificate, which impacts their complexity.
The complexity of Lipshitz's algorithm is~$d^{\smash{O(n^2)}}$ operations in~$k$; the complexity of no other algorithm is known.
Pancratz~\cite{Pan10} developed an approach similar to ours, under a restrictive hypothesis, much stronger than Griffiths' regularity assumption.
He proceeds to a complexity analysis of his algorithm but in terms of operations in~$k(t)$ rather than in the base field~$k$.
Algorithms based on non-commutative Gröbner bases and elimination~\cite{Zei90,Takayama1990a} or based on the search of rational solutions to differential equations~\cite{Chy00} resist complexity analysis.
The method of Apagodu and Zeilberger~\cite{ApaZei06} requires a generic exponent and %
specialization seems problematic.

For the restricted class of diagonals of rational functions, there is a heuristic based on series expansion and differential approximation~\cite{KauZei11}; it does not need to compute a certificate.
However, even using the bounds in~$d^{\smash{O(n)}}$, its direct implementation has a complexity of~$d^{\smash{O(n^2)}}$ operations in~$k$.

\medskip
\noindent\emph{Contributions.}
Our main result, obtained with the \GD\ method and a deformation technique, is the existence of a telescoper with regular certificate of order at most $d^n$ and degree~$d^{O(n)}$ that can be computed in~$d^{O(n)}$ arithmetic operations in~$k$.
For generic homogeneous rational functions, the telescoper computed is the minimal order telescoper with regular certificate.
Theorems~\ref{coro:asy-deg},~\ref{thm:complexity-reg} and~\ref{thm:bound-affine} state precise complexity and size estimates. 
To the best of our knowledge, the bounds on the order and degree are better than what was known and it is the first time that a complexity single exponential in~$n$ is reached.
For a generic rational function, every pair (telescoper, regular certificate) has a size larger than~$d^{O(n)}$, see Remark~\ref{rem:size-cert}, but our algorithm does not need to compute the certificate.
A prototype implementation shows that this algorithm can lead
to a spectacular improvement over previous methods, though the domain of improvement is not  satisfactory yet.

\medskip
\noindent\emph{Acknowledgement.}
We are grateful to G.~Christol for
many rewarding discussions, and we thank G.~Villard and W.~Zhou for communicating  their complexity results in linear algebra.

\section{Overview of the method}

\noindent In this section we introduce the basics of the Griffiths--Dwork method.
In dimension~1, this method coincides with classical Hermite reduction, which we first recall.

\subsection{Dimension one: Hermite reduction}\label{sec:hermite}

\noindent Let~$F$ be a rational function in~$x$, over a field~$L$, written as~$a/f^\ell$, with~$a$ and~$f$ two polynomials not necessarily coprime, the latter being square-free, \emph{i.e.} the polynomials~$\partial_x f$ and~$f$ are coprime.
In particular~$a$ equals~$uf+v \partial_x f$ for some polynomials~$u$ and~$v$.
Then, if~$\ell>1$, the function $F$ rewrites 
\[ F  = \frac{u+ \tfrac{1}{\ell-1} \partial_x v}{f^{\ell-1}}+\partial_x\left( \frac{-v}{(\ell-1)f^{\ell-1}} \right). \]
Iterating this reduction step $\ell$ times gives~$F$ as~$\frac{U}{f}+\partial_x\frac{V}{f^{\ell-1}}$
for some polynomials~$U$ and~$V$.
Next, Euclidean division allows to write~$U$ as~$r+sf$, with~$r$ of degree less than the degree of~$f$, yielding the additive decomposition
\[ F = \frac{r}{f} + \partial_x\left(\frac{V}{f^{\ell-1}}+ {\textstyle \int}s \right). \]
The rational function~$r/f$ is the \emph{reduced form} of~$F$ and is denoted by~$[F]$.
This form features important properties:
\begin{description}
  \item[\normalfont\emph{(Linearity)}] $f$ being fixed, $[F]$ depends linearly on~$F$;
  \item[\normalfont\emph{(Soundness)}] if $[F]$ is zero, then~$F$ is a derivative w.r.t.~$x$;
  \item[\normalfont\emph{(Confinement)}] $[F]$ lies in a finite-dimensional vector space over~$L$ depending only on~$f$ (with dimension~$\deg_x f$);
  \item[\normalfont\emph{(Normalization)}] if~$F$ is a derivative w.r.t.~$x$, then~$[F]$ is zero.
\end{description}
These properties are enough to compute a telescoper:
Assume now that~$L$ is~$k(t)$ for a field~$k$.
If for some elements of~$L$, say~$a_0,\dotsc,a_p$, the reduced form~$\left[\sum_i a_i \partial^i_t F\right]$ vanishes, then the operator~$\sum_i a_i\partial_t^i$ is a telescoper, thanks to the soundness property.
Thanks to the linearity property, this is equivalent to the vanishing of~$\sum_i a_i \left[\partial^i_t F\right]$.
Thanks to the confinement property, it is always possible to find such a relation.
Thanks to the normalization property, every telescoper arises in this way. In particular, so does the  telescoper of minimal order.

\subsection{Griffiths--Dwork reduction}
\noindent %
Let~$F$ be a rational function in~$n$ variables~$x_1,\dotsc,x_n$, written as~$a/f^\ell$, with~$f$ a square-free polynomial.
If~$\ell>1$ and if~$a$ lies in the ideal of~$L[x_1,\dots,x_n]$ generated by~$f$ and its derivatives~$\partial_i f$,
then we can write~$a$ as~$uf+\sum_iv_i\partial_if$,
for some polynomials $u,v_1,\ldots,v_n$, and~$F$ rewrites
\[
  F=  \frac{u+\tfrac{1}{\ell-1}\sum_{i=1}^n \partial_i v_i}{f^{\ell-1}} + \sum_{i=1}^n\partial_i\left( \frac{-v_i}{(\ell-1)f^{\ell-1}} \right).
\]
Provided that  this ideal %
contains~1,
any~$F$ can be reduced to a function with simple poles by iteration of this identity.
The soundness and linearity properties are naturally satisfied, but extending further the reduction
to obtain at least the confinement property is not straightforward and requires stronger assumptions~\cite[\S 4]{Mov11a}.
A difficulty with this approach is that the degrees of the cofactors~$v_i$ at each reduction step are poorly controlled: we lack the Euclidean division step and we reduce poles at finite distance at the cost of making worse the pole at infinity. This difficulty is overcome by working in the projective space. %
The translation between affine and projective is discussed more precisely in Section~\ref{sec:singular}.

Now, assume that~$a$ and~$f$ are \emph{homogeneous} polynomials in~$\Lx=L[x_0, \dots,x_n]$, with~$f$ of degree~$d$.
A central role is played by the \emph{Jacobian ideal $\Jac f$ of~$f$}, the ideal generated by the partial derivatives $\partial_0f,\dots,\partial_n f$.
Note that since~$f$ is homogeneous, Euler's relation, which asserts that~$f$ equals~$\frac1d  \sum_{i=0}^n x_i\partial_if$
implies that~$f\in\Jac f$. 

We now decompose~$a$ as~$r + \sum_i v_i \partial_i f$. In contrast with the
affine case, each nonzero~$v_i$ can be chosen homogeneous of degree precisely~$\deg a - \deg\partial_i f$.
If~$\ell>1$, we obtain
\begin{equation}\label{eqn:reduce-identity}
  F = \frac{r}{f^\ell} + \underbrace{\frac{\tfrac{1}{\ell-1}\sum_{i=0}^n \partial_i v_i}{f^{\ell-1}}}_{F_1}+\sum_{i=0}^n\partial_i\left( \frac{-v_i}{(\ell-1)f^{\ell-1}} \right).
\end{equation}
If~$r$ is not zero, the order of the pole need not decrease, contrary to the affine case,  
but~$r$ is reduced to a normal form modulo~$\Jac f$; this will help us obtain the confinement property, see Proposition~\ref{coro:hbasis}.
The reduction process proceeds recursively on~$F_1$,  which has pole order~$\ell-1$, and stops when~$\ell=1$.
This procedure is summarized in Algorithm~\ref{algo:reduce}.

\begin{algo}
\centering
\begin{algoenv}{
    $F = a/f^\ell$ a rational function in~$x_0,\dotsc,x_n$
  }{$[F]$ such that there exist rational functions~$A_0,\dotsc,A_n$ such that $F = [F] +\sum_i\partial_i A_i$}
  \State Precompute a Gr\"obner basis~$G$ for~$(\partial_0f,\dotsc,\partial_nf)$
  \Procedure{Reduce}{$a/f^\ell$}
  \If{$\ell=1$} \Return $a/f^\ell$
  \EndIf
  \State Decompose~$a$ as~$r + \sum_i v_i \partial_i f$ using~$G$
  \State $F_1 \gets \frac{1}{\ell-1}\sum_i\frac{\partial_iv_i}{f^{\ell-1}}$
  \State \Return $\frac{r}{f^\ell}+\Reduce(F_1)$
  \EndProcedure
\end{algoenv}
\caption{Griffiths--Dwork reduction}
\label{algo:reduce}
\end{algo}

\section{Properties of the Griffiths--\\Dwork reduction}
\label{sec:prop-grif}
\noindent Let~$f$ in~$L[\xx]$ be a homogeneous polynomial of degree~$d$, where~$L$ is a field of characteristic zero.
It is clear that the reduction procedure satisfies the soundness and the linearity properties.
Analogues of confinement and normalization hold under the following regularity hypothesis: 
\[ \tag{H}\label{eqn:regularity} \Lx/\!\Jac f\text{ is finite-dimensional over $L$}. \]
Geometrically, this hypothesis means that the hypersurface defined by~$f$ in~$\PP^n$ is smooth.
In particular~$f$ is irreducible.

The ring of rational functions in $L(\xx)$ whose denominator is a power of~$f$ is denoted by~$L[\xx,\tfrac{1}{f}]$.
Let~$\Lxf_p$ denote the subspace of homogeneous functions of degree~$p$, \emph{i.e.} the set of~$F$ in~$L[\xx,\tfrac{1}{f}]$ such that~$F(\lambda \xx)$ equals~$\lambda^p F(\xx)$.
Note that each derivation~$\partial_i$ induces a map from~$\Lxf_p$ to~$\Lxf_{p-1}$.
Let $D_f$ denote the subspace of~$L[\xx,\tfrac{1}{f}]$ consisting of rational functions~$\sum_i \partial_i A_i$ for some~$A_i$ in~$L[\xx,\tfrac{1}{f}]_{-n}$.
A major character of this study is the quotient space~$\Lxfp/D_f$, denoted by~$\Hfp$.

The reduced form of~$F$ in~$\Lxfp$ is denoted by~$[F]$.
It is by definition the output of the algorithm~\tReduce.
It depends on a choice of a Gröbner basis of~$\Jac f$, but its vanishing does not, see Theorem~\ref{thm:griffiths} below.

The choice of the space~$\Lxfp$ and the degree~$-n-1$ may seem arbitrary. It is motivated by it being isomorphic to the space of regular differential~$n$-forms on~$\PP^n\setminus V(f)$.
The evaluation of~$x_0$ to~$1$ is the restriction map to~$\mathbb{A}^n\setminus V(f)$.
The space~$\Hfp$ is the~$n$th de~Rham cohomology space of the algebraic variety~$\PP^n\setminus V(f)$ over~$L$.

\begin{thm}[Griffiths~{\cite[\S4]{Gri69}}]\label{thm:griffiths}
  \hspace{.5em} %
  If~$f$ satisfies Hypothesis~\eqref{eqn:regularity},
   then for all~$F$ in~$\Lxfp$,
 the reduced form~$[F]$ vanishes if and only if~$F$ is in~$D_f$.
\end{thm}

Theorem~\ref{thm:griffiths} gives access to the dimension of~$\Hfp$.
Let~$A$ be the finite dimensional vector space~$\Lx/\!\Jac f$.
For a positive integer~$l$, let~$A_l$ denote the linear subspace of~$A$ generated by homogeneous polynomials of degree~$ld-(n+1)$. 
Let~$B$ denote~$\oplus_l A_l$.
Finally, for~$l>0$ let~$(g_{l,i})_{1\leqslant i \leqslant n_l}$ be a basis of~$A_l$, with~$n_l = \dim_L A_l$. 

\begin{prop}\label{coro:hbasis}
Under Hypothesis~\eqref{eqn:regularity}, the family of rational functions~$\left(g_{l,i}/f^\ell\right)_{0<l,i\leqslant n_l}$ induces a basis of~$\Hfp$.
\end{prop}

\begin{proof}
Suppose there exists a linear relation between the~$g_{l,i}/f^\ell$ modulo~$D_f$, that is~$\sum_{l,i} u_{l,i}g_{l,i}/f^\ell$, denoted by~$F$, lies in~$D_f$ for some elements~$u_{l,i}$ of~$L$, not all zero.
Let~$l_0$ be the maximum~$l$ such that~$u_{l,i}$ is not zero for at least one~$i$.
By Theorem~\ref{thm:griffiths}, $[F]=0$ so that 
$\sum_{l,i}u_{l,i}g_{l,i}f^{\ell_0-l}$, the numerator of~$F$, lies in~$\Jac f$.
Since~$f$ itself is in~$\Jac f$, so is the sum~$\sum_i u_{l_0, i}g_{l_0, i}$, contradicting the fact that the~$g_{l_0,i}$ are a basis of~$A_{l_0}$. 
Thus the~$g_{l,i}/f^\ell$ form a free family.

To prove that this family generates~$\Hfp$, we first notice that the family of all the fractions~$[F]$, for~$F$ in~$L[\xx,\tfrac{1}{f}]_{-n-1}$, generates~$\Hfp$ since~$[F]$ equals~$F$ modulo~$D_f$.
Now we assume for a moment that each~$g_{l,i}$ is reduced with respect to a Gr\"obner basis~$G$ of~$\Jac f$.
Then each polynomial of~$\Lx$ of degree~$ld-n-1$ which is reduced with respect to~$G$ is a linear combination of the~$g_{l,i}$.
Thus for all~$F=a/f^\ell$ in~$L[\xx,\tfrac{1}{f}]_{-n-1}$, the reduction~$[F]$ is in the span of all the~$g_{l,i}/f^\ell$. %
This makes the~$g_{l,i}/f^\ell$ a system of generators of~$\Hfp$ and by the previous paragraph a basis of it. Thus~$\Hfp$ has the same dimension as~$B$ and
any free family of~$\Hfp$ of cardinal~$\dim_L B$ is a basis of~$\Hfp$.
In particular, the~$g_{l,i}/f^\ell$ form a basis even if the~$g_{l,i}$ are not reduced with respect to~$G$.
\end{proof}

\begin{coro}\label{coro:hdim}
  Under Hypothesis~\eqref{eqn:regularity},~$\Hfp$ has dimension
  \[ \frac 1d \left((d-1)^{n+1}+(-1)^{n+1}(d-1)\right) \quad\mathrlap{\big( \leqslant d^n \big).} \]
\end{coro}

\begin{proof}
It has the dimension of~$B$, see~\cite[thm. 8.3]{Mon70} for its computation. The inequality is clear.
\end{proof}

\section{Creative Telescoping}

\noindent We now introduce an algorithm, based on the Griffiths--Dwork reduction, that computes a telescoper of a rational function under Hypothesis~\eqref{eqn:regularity}.

In Equation~\eqref{eqn:prob-ct}, the telescoper~$T$ is said to have a \emph{regular certificate} if the irreducible factors of the denominators of the~$A_i$'s, as rational functions over~$k(t)$, divide the denominator of~$F$; in other words, the~$A_i$'s have no pole outside those of~$F$, over~$k(t)$. 
Algorithm~\ref{algo:ct}, described in~\S\ref{sec:ct-algo}, returns the telescoper of minimal order having regular certificate.
For the application of creative telescoping to integration, this class of telescopers is more interesting than the general one; that is the object of~\S\ref{sec:regular-certificate}.

\subsection{Telescopers with regular certificate}\label{sec:regular-certificate}

\noindent Back to the affine case, let~$F(t,x_1,\dotsc,x_n)$ be a rational function over~$\CC$ and~$\gamma$ be a~$n$-cycle in~$\CC^n$ over which~$F$ has no pole for a generic~$t$ in~$\CC$. %
A common use of creative telescoping is the computation of a differential equation satisfied by the one-parameter integral~$I(t)=\int_{\smash{\gamma}} F\ud \xx$.
As mentioned in the introduction, it is not always possible to deduce from the telescoping equation~\eqref{eqn:prob-ct}
that~$T(I)$ vanishes. It may happen that the polar locus of the certificate meets~$\gamma$ for all~$t \in\CC$, and so some~$\int_{\smash{\gamma}}\partial_i A_i\ud \xx$ need not be zero.
An example of this phenomenon is given by Picard~\cite{Pic99} for a bivariate algebraic function
and translated here into a rational example, using the method in~\cite[Lemma 4]{CheKauSin12}:
\begin{multline}\label{eqn:picard-ex}
  \frac{x-y}{z^2-P_t(x)P_t(y)} = \partial_x\tfrac{2 P_t(x)}{(x-y) \left(z^2-P_t(x)P_t(y)\right)} +\\
   \partial_y \tfrac{2 P_t(y)}{(x-y) \left(z^2-P_t(x)P_t(y)\right)}+\partial_z \tfrac{3 \left(x^2+y^2\right) z}{(x-y) \left(z^2-P_t(x)P_t(y)\right)} , 
\end{multline}
where~$P_t(u)=u^3+t$.
Note the factor~$x-y$ in the denominator of the certificate.
The operator~$1$ is a telescoper of the left-hand side~$F$, however
there exists a $3$-cycle~$\gamma$ on which~$F$ has no pole and such that~$\int_\gamma F\ud \xx$ is not zero. 
It is thus impossible to find a regular certificate for the telescoper~$1$.

Nevertheless, a differential equation for~$I(t)$ can be obtained in two ways.
First, one can carefully study the integral~$\sum_i\int_\gamma \partial_i A_i \ud \xx$ and compute a differential equation for it.
Usually this includes the analysis of the poles of the~$A_i$'s, and the search of a telescoper for some rational function with one variable less.
The second way is to find a telescoper for~$F$ such that the certificate does not contain new poles, a \emph{telescoper with regular certificate}.
Contrary to the telescoper~\eqref{eqn:picard-ex}, the operator~$\partial_t$ is a telescoper with regular certificate:
\[ \partial_t F = \partial_x\left( -\tfrac{x}{3t}F \right)+\partial_y\left(- \tfrac{y}{3t}F \right)+\partial_z\left(- \tfrac{z}{t}F \right). \]
This proves that~$\partial_t I = 0$. More generally we have:
\begin{prop}
  If~$T\in\CC(t)\langle \partial_t \rangle$ is a telescoper of~$F$ with regular certificate, then~$T(I)$ is zero.
\end{prop}
In this case, the certificate itself is not needed to prove the conclusion, its existence and regularity are sufficient. The Griffiths--Dwork method always produces a telescoper with regular certificate, see Equation~\eqref{eqn:reduce-identity}.

\subsection{Algorithm}\label{sec:ct-algo}
\begin{algo}
\centering
\begin{algoenv}{
    $F = a/f^\ell$ a rational function in~$\Lxfp$, with~$f$ satisfying~\eqref{eqn:regularity}
  }{$T(t,\partial_t)$ an operator such that~$T(F) = \sum_i\partial_iA_i$ for some rational functions~$A_i$}
  \Procedure{\tCreatTel}{$F$}
    \State $G_0\gets \Reduce(F)$ 
    \State $i\gets 0$
    \Loop
    \If{$\mop{rank}_L(G_0,\dotsc,G_i) < i+1$}
      \State solve~$\sum_{k=0}^{i-1} a_k G_k = G_i$ w.r.t.~$a_0, \dotsc, a_{i-1}$ in~$L$
      \State \Return $\partial_t^i - \sum_k a_k\partial_t^k$
    \Else
      \State $G_{i+1} \gets \Reduce(\partial_t G_{i})$
      \State $i\gets i+1$
    \EndIf
    \EndLoop
  \EndProcedure
\end{algoenv}
\caption{Creative telescoping, regular case}
\label{algo:ct}
\end{algo}

\noindent In this section~$L$ is~$k(t)$ for some field~$k$ and $f$ is a homogeneous polynomial over~$L$ of degree~$d$ satisfying Hypothesis~\eqref{eqn:regularity}. 
For~$F$ a rational function in~$\Lxfp$ we want to find a nonzero operator~$T$ in~$L\langle \partial_t \rangle$ such that~$T(F)$ lies in~$D_f$.
Algorithm~\ref{algo:ct} describes the procedure \tCreatTel\ that outputs such a telescoper.
Note that~$\Lxfp$ is stable with respect to the derivation~$\partial_t$.

\begin{prop}
Algorithm~\ref{algo:ct} terminates and outputs the minimal telescoper of~$F$ that has regular certificate.
\end{prop}

\begin{proof}
The sequence~$(G_k)$ is defined by~$G_0 = [F]$ and the recurrence relation~$G_{k+1} = [\partial_t G_k]$.
We show by induction that for all~$k$ the fraction~$G_k$ equals~$[ \partial_t^k F]$.
This is clear for~$k=0$. Assume that~$G_k$ equals~$[ \partial_t^k F]$.
By the soundness of the reduction the operator~$G_k  - \partial^k_t F$ lies in~$D_f$. 
And then so does~$\partial_t G_k - \partial_t^{k+1} F$ since~$\partial_t$ commutes with the~$\partial_i$'s.
By Theorem~\ref{thm:griffiths} and linearity, this implies that~$[\partial_t G_k]$ equals~$[ \partial_t^{k+1} F]$.

At the~$i$th step of the loop the algorithm is looking for a linear relation between~$[F],\dotsc,[\partial^i_t F]$.
By Theorem~\ref{thm:griffiths}, there is one if and only if there is a telescoper with regular certificate of order~$i$.
If there is such a relation, the algorithm computes it and returns the corresponding telescoper.
By %
Proposition~\ref{coro:hbasis}, the algorithm terminates.
The telescoper admits a regular certificate by design, see Equation~\eqref{eqn:reduce-identity}.
\end{proof}

\section{Effective bounds for creative telescoping}\label{sec:eff-bounds}

\noindent We now review the steps of the algorithm with the aim of bounding the degrees and orders of all polynomials and operators that are constructed.
This is then used in the next section to assess the complexity of this approach.

For the needs of Section~\ref{sec:singular}, we track the degrees not only with respect to the parameter~$t$ but also to another free variable~$\varepsilon$ of the base field. 
In other words, we assume that~$L$ is~$k(t,\varepsilon)$.
For~$p$ a polynomial in~$k[t,\varepsilon]$, the bi\-degree~$(\deg_t p, \deg_\varepsilon p)$ of~$p$ is denoted by~$\delta(p)$.
If~$p=\sum_I p_I \mathbf x^I$ is a polynomial in~$t$,~$\varepsilon$ and~$\mathbf x$, then~$\delta(p)$ denotes the supremum of the~$\delta(p_I)$'s, component by component.

\begin{thm}\label{coro:asy-deg}
Let $f\in L[\xx]$ be homogeneous of degree~$d$ satisfying~\eqref{eqn:regularity}.
  Let~$a/f^\ell$ in~$\Lxfp$ be a rational function, with~$a$ a polynomial in~$t$ and~$\varepsilon$.
  The minimal telescoper  of~$a/f^\ell$  with regular certificate has  degree
  \[ \bigO \left( d^n \delta(a) + \left( ld^{2n}+ d^{3n} \right)\eexp^{n} \delta(f) \right), \]
  uniformly in all the parameters. It has order at most~$d^n$.
\end{thm}
The last part of the theorem is a direct consequence of the confinement property of Corollary~\ref{coro:hdim}.
We now study more precisely the decomposition used in Algorithm~\ref{algo:reduce} in order to control the degree of the telescoper and complete the proof.

The notation~$a(\mathbf n) = \bigO(b(\mathbf n))$, for a tuple~$\mathbf n$, means that there exists~$C>0$ such that for all~$\mathbf n \geqslant 1$, with at most a finite number of exceptions, we have~$a(\mathbf n) \leqslant Cb(\mathbf n)$.
The notation~$a(\mathbf n) = \softO(b(\mathbf n))$ means that~$a(\mathbf n) = \bigO(b(\mathbf n) \log^k b(\mathbf n))$ for some integer~$k$.
We emphasize that when there are several parameters in a~$\bigO$, the constant is uniform in all the parameters and there is at most a finite number of exceptions.

\subsection{Reduction modulo the Jacobian ideal}
\label{sec:red-mod-jac}

\noindent An important ingredient of the Griffiths--Dwork reduction is the computation of a
decomposition~$r+\sum_i u_i\partial_if$ of a homogeneous polynomial~$a$.  This
can be done by means of a Gröbner basis of~$\Jac f$, but instead of following
the steps of a Gr\"obner basis algorithm, we cast the computation into a linear
algebra framework using Macaulay's matrices, for which Cramer's rule and
Hadamard's bound can then be used.
While not strictly equivalent, both methods ensure that~$r$ depends linearly on~$a$ and vanishes when~$a$ is in~$\Jac f$.

For a positive integer~$q$, let~$\varphi_q$ denote the linear map
\[ \varphi_q :  (u_i) \in \Lx_{q-d-n}^{n+1} \longrightarrow \sum_{i=0}^{\smash{n}} u_i\partial_i f \in \Lx_{q-n-1}. \]
Let~$\Mat \varphi_q$ be the matrix of~$\varphi_q$ in a monomial basis.
It has dimension~$R_q \times C_q$, where~$R_{q}$ denotes~$\tbinom{q-1}{n}$ and~$C_{q}$ denotes~$(n+1)\tbinom{q-d}{n}$,
and we note for future use that~$C_{q}\leqslant R_{q}$ for all positive integers~$n$ and~$d>2$.
Up to a change of ordering of the bases of the domain and codomain, %
$\Mat \varphi_q$ has the form $\left(\begin{smallmatrix} A & B \\ C & D \end{smallmatrix}\right)$,
where~$A$ is a  square submatrix of maximal rank.
Note that~$D$ is necessarily~$CA^{-1}B$.
Then, the endomorphism~$\psi_q$ defined by the matrix~$\left(\begin{smallmatrix} A^{-1} & 0 \\ 0 & 0 \end{smallmatrix}\right)$
satisfies~$\varphi_q \psi_q \varphi_q = \varphi_q$; it is called a \emph{split} of $\varphi_q$.
It depends on the choice of the maximal rank minor.
The map~$\id - \varphi_q \psi_q$, denoted by~$\pi_q$, performs the reduction in degree~$q-n-1$: it is idempotent;
if~$a$ of degree $q-n-1$ is in~$\Jac f$ then it equals~$\varphi_q(b)$ for some~$b$ and thus~$\pi_q(a)$ vanishes;
and for all~$a$ in~$\Lx_{q-n-1}$ it gives a decomposition
\[ a = \pi_q(a) + \sum_i \psi_q(a)_i \partial_i f. \]

Under Hypothesis~\eqref{eqn:regularity}, the map~$\varphi_q$ is surjective when~$q$ is at least~$(n+1)d - n$.
Let~$D$ denote this bound, known as \emph{Macaulay's bound}~[\citen{Mac16}, chap.~1; \citen{Laz77}, corollaire, p.~169]. 

For~$q$ larger than~$D$, a split of~$\varphi_q$ can be obtained from a split~$\psi_{D}$ of~$\varphi_{D}$ in the following way.
Let~$S$ be the set of monomials in~$\xx$ of total degree~$q-D$.
Choose a linear map $\mu$ from~$\Lx_{q-n-1}$ to~$\Lx_{D-n-1}^S$ such that each~$a$ in~$\Lx_{q-n-1}$ equals~$\sum_{m\in S} m \mu_m(a)$.
Then a split of~$\varphi_d$ is defined by
\[ \psi_d(a) = \sum_{m\in S} m \psi_{D}(\mu_m(a)).\]

Let~$q$ be a positive integer and let~$E_q$ be the least common multiple of the denominators of the entries of~$\Mat \psi_q$.
The entries of~$\Mat \psi_q$ and~$\Mat \pi_q$ are rational functions of the form~${p}/{E_q}$, with~$p$ polynomial.
Let~$\delta_E$ denote the supremum of all~$\delta(p)$ and all~$\delta(E_q)$, for $q\in\mathbb{N}\setminus\{0\}$.

\begin{prop}\label{prop:deltaE}
  The supremum~$\delta_E$ is finite and bounded above by~$\eexp^n d^n \delta(f)$. 
  Moreover, if~$q>D$ then~$E_q$ equals~$E_{D}$.
\end{prop}

\begin{proof}
  Assume first that~$q > D$.
  In this case, the entries of~$\Mat \psi_q$ are entries of~$\Mat\psi_{D}$ and~$\pi_q$ is zero.
  Thus the inequalities will follow from the case where~$q\leqslant D$. 
  Let~$\Mat \psi_q$ and~$\Mat \pi_q$ be written respectively as~$N/E_q$ and~$P/E_q$ with~$N$ and~$P$ polynomial matrices.
  Let~$r$ be the rank of~$\psi_q$.
  The maximal rank minor~$A$ in the construction of~$\psi_q$ has dimension~$r$.
  Cramer's rule and Hadamard's bound ensure that~$\delta(N)$ is at most~$(r-1)\delta(f)$ and that~$\delta(E_q)$ is at most~$r\delta(f)$.
  Since~$P$ equals~$E_q\id - (\Mat\varphi_q) N$ and~$\delta(\Mat\varphi_d)$ equals~$\delta(f)$, the degree~$\delta(P)$ is also at most~$r\delta(f)$.

Next, $r$ is bounded by~$R_{q}$, the row dimension of~$\Mat\phi_d$.
Since~$q\leqslant D$, we have~$R_{q}\leqslant R_{D}$ and we conclude using  inequality 
$\tbinom{p}{n}\leqslant\left(\tfrac{p\,\eexp}{n+1}\right)^n$, with~$p\geqslant n$ an integer.
\end{proof}

Algorithm~\ref{algo:reduce/linalg} is a slightly modified version of Algorithm~\ref{algo:reduce} which uses 
the construction above. Its output is in general not equal to the output of the former version, for any monomial order, but of course it satisfies Theorem~\ref{thm:griffiths}. In particular the output of the algorithm~\tCreatTel\ does not depend on the reduction method in~\tReduce.
From now on the brackets~$[\cdot]$ denote the output of Algorithm~\ref{algo:reduce/linalg}.

\subsection{Degree bounds for the reduction}

\begin{prop}\label{prop:bound-denom}
Let~$a/f^l\in\Lxfp$, with~$a$ a polynomial in~$t$ and~$\varepsilon$. Then
\[ \left[\frac{a}{f^\ell}\right] = \frac{1}{P_\ell}\sum_{k=1}^{n} \frac{b_k}{f^k}, \]
where~$P_\ell=\prod_{i=1}^\ell E_{i d}$ and $b_k$ in~$\Lx_{kd-n-1}$ is a polynomial in~$t$ and~$\varepsilon$ such that~$\delta(b_k)\leqslant \delta(a)+\ell\delta_E$, for~$1\le k\le n$.
\end{prop}

\begin{proof}
Using Algorithm~\ref{algo:reduce/linalg}, we obtain
  \[ \left[\frac{a}{f^\ell}\right] = \frac{p}{E_{\ell d} f^\ell}+\frac{1}{E_{\ell d}}\left[\frac{g}{f^{\ell-1}}\right], \]
where~$g$ and~$p$ are polynomials in~$\xx$, $t$ and~$\varepsilon$, with~$\delta(p)$ and~$\delta(g)$ at most~$\delta(a)+\delta_E$.
Induction over~$l$ yields
  \[ \left[ \frac{a}{f^l} \right]
    = \sum_{k=1}^l \frac{p_{k}}{f^{k} \prod_{j=k}^l E_{jd}}
  \]
with~$p_k$ polynomials such that~$\delta(p_k)\leqslant \delta(a)+(l-k+1)\delta_E$.
  For~$k>n$, and hence~$kd>D$, the map~$\pi_{k d}$ is~0 and thus so is~$p_{k}$.
  Thus
  \[ \left[ \frac{a}{f^l} \right] = \frac{1}{\prod_{j=1}^{l}E_{jd}} \sum_{k=1}^{\min(l,n)} \frac{p_k\prod_{j=1}^{k-1}E_{jd}}{f^k}. \qed \]
\end{proof}

\begin{algo}
\centering
\begin{algoenv}{
$F = a/f^\ell$ a rational function in~$\Lxfp$, with~$f$ of degree~$d$ %
  }{$[F]$ such that there exist rational functions~$A_0,\dotsc,A_n$ such that $F = [F] +\sum_i\partial_i A_i$}
  \State For all~$1\leqslant i \leqslant \ell$, precompute a split~$\psi_{i d}$ of~$\varphi_{i d}$ (\S\ref{sec:red-mod-jac})
  \Procedure{Reduce}{$a/f^\ell$}
  \If{$\ell=1$}
\Return $a/f^\ell$
  \EndIf
  \State $\displaystyle F_1 \gets \frac{1}{\ell-1}\sum_i\frac{\partial_i \psi_{\ell d}(a)_i}{f^{\ell-1}}$
 
  \State \Return $\frac{\pi_{\ell d}(a)}{f^\ell}+ \Reduce(F_1)$
   \EndProcedure

\end{algoenv}
\caption{Griffiths--Dwork reduction, linear algebra variant}
\label{algo:reduce/linalg}
\end{algo}

This proposition applied to~$\partial_t^i(a/f^l)$ asserts that
\begin{equation}\label{bound:deg-diff}
\left[ \partial_t^i\frac{a}{f^\ell} \right] = \frac{1}{P_{l+i}}\sum_{k=1}^n\frac{b_{i,k}}{f^k} = \frac{1}{P_{l+i}}\frac{b'_{i}}{f^n}
\end{equation}
for some polynomials~$b_{i,k}$ and~$b'_i$ such that
\begin{align}
  \label{eqn:deltabk} \delta(b_{i,k}) &\leqslant \delta(a) + i \delta(f) + (i+l)\delta_E, \\
  \label{eqn:deltabk'} \text{and}\quad\delta(b'_i) &\leqslant \delta(a) + (i+n) \delta(f) + (i+l)\delta_E.
\end{align}

\subsection{Degree bounds for the telescoper}
\begin{prop}\label{thm:delta-ck}
  Let~$T=\sum_{i=0}^{r}c_i \partial_t^i$, with coefficients~$c_i$ in~$k[t,\varepsilon]$, be the minimal telescoper with regular certificate of~$a/f^l$.
  Then
  \begin{align*}
  \delta(c_i) &\leqslant r\delta(a) + \left( r^2 + rl  \right)  \eexp^{n} d^n\delta(f).
  \end{align*}
\end{prop}

\begin{proof}
The operator~$T$ is the output of~$\CreatTel(a/f^l)$.
The rational functions~$c_i/c_r$ form the unique solution to the following system of inhomogeneous linear equations over~$L$, with the~$Y_i$'s as unknown variables:
\[ \sum_{i=0}^{r-1} \left[ \partial_t^i\frac{a}{f^\ell} \right] Y_i = - \left[ \partial_t^r\frac{a}{f^\ell} \right]. \]
We write each~$b_{i,k}$ in~\eqref{bound:deg-diff} as~$\sum_{m\in S}b_{i,k,m} m$, where~$S$ is the set of all monomials in the variables~$\xx$ of degree at most~$nd-n-1$.
The previous system rewrites as
\[  \forall m\in S, \forall k \in \left\{ 1,\dots,n \right\},\quad \sum_{i=0}^{\smash{r-1}} Y_i \frac{b_{i,k,m}}{P_{l+i}} = -\frac{b_{r,k,m}}{P_{l+r}} \]
There is a set~$I$ of~$r$ indices~$\{ (k_0,m_0),\dotsc \}$ such that the square system formed by the corresponding equations admits a unique solution.
We apply Cramer's rule to this system.
Let~$B$ be the square matrix~$(b_{i,k_j,m_j})_{i,j}$, %
for $0\le i,j<r$.
Let~$B_i$ be the matrix obtained by replacing the row number~$i$ of~$B$ by the vector~$(b_{r,k_j,m_j})_j$.
We get, after simplification of the factors~$P_{l+*}$ by multilinearity of the determinant,
\begin{equation}\label{eqn:cr-div-det}
  \frac{c_i}{c_r} =\frac{ \tfrac{P_{l+i}}{P_l} \det B_i }{\tfrac{P_{l+r}}{P_l} \det B}.
\end{equation}
So, for all~$i$, the polynomial~$c_i$ divides~$\tfrac{P_{l+i}}{P_l} \det B_i$ and thus
\begin{equation*}
  \delta(c_i) \leqslant i\delta_E + \sum_{j=0, j\neq i}^{\smash{r}} \delta(b_j).
\end{equation*}
With the previous bound~\eqref{eqn:deltabk} on~$\delta(b_i)$ we get
\[  \delta(c_i) \leqslant r\delta(a) + \frac{r(r+1)}{2}\left( \delta(f)+\delta_E \right) + rl \delta_E, \]
which gives the result with Proposition~\ref{prop:deltaE}.
\end{proof}

\section{Complexity}

\begin{table*}[bt]
  \centering
  \begin{tabular}{c *{4}{@{\hspace{2em}}c@{}>{ \itshape(}c<{) }@{ }c}}
    \toprule
    degree of $f$  & \multicolumn{3}{c}{3} & \multicolumn{3}{c}{4} & \multicolumn{3}{c}{5} & \multicolumn{3}{c}{6} \\\midrule
     order of telesc. & \multicolumn{3}{c}{2} & \multicolumn{3}{c}{6} & \multicolumn{3}{c}{12} & \multicolumn{3}{c}{20} \\\midrule
     degree of telesc. $\delta = 1$ &
        32 & 68 & 0.4s &
        153 & 891 & 46s &
        480 & 5598 & 2h &
        1175 & 23180 & 150h\\
     \phantom{degree of telesc}\llap{---\quad}, $\delta = 2$ & %
        66 &    136   & 0.6s &
        336 &   1782  & 140s &
        1092 &  11196 & 7h    &
        ? &     46360 & $\varnothing$ \\
     \phantom{degree of telesc}\llap{---\quad}, $\delta = 3$ & %
        100 &   204 &   0.9s  &
        519 &   2673 &  270s  &
        1704 &  16794 & 13h   &
        ? &     69540 & $\varnothing$ \\\bottomrule
  \end{tabular}
  \caption{Empirical order and degree of the minimal telescoper with regular certificate of a random rational function~$a/f^2$ in~$\QQ(t,x_0,x_1,x_2)$, with~$f$ and~$a$ homogeneous in~$\xx$ satisfying $\deg_\xx a + 3 = 2\deg_\xx f$ and~$\delta(a)$ and~$\delta(f)$ equal to~$\delta$; together with a proved upper bound (with a version of Theorem~\ref{thm:delta-ck} without simplification) and mean computation time (CPU time).}

  \label{tab:orddeg-exp}
\end{table*}

\noindent We assume that~$L$ is the field~$k(t)$ and 
we evaluate the algebraic complexity of the steps of \tReduce\ and \tCreatTel\ in terms of number of arithmetic operations in~$k$.
All the algorithms are deterministic.
For univariate polynomial computations, we use the quasi-optimal algorithms in~\cite{GatGer03}.
For simplicity, we assume that~$d>2$ so that several simplifications occur in the inequalities since~$C_q\leqslant R_q$ and $d>\eexp \approx 2.72$.
\begin{thm}\label{thm:complexity-reg}
Under Hypothesis~\eqref{eqn:regularity} and assuming that $d>2$, Algorithm~\tCreatTel\ run with input $a/f^\ell$ takes
  \begin{equation*}
    \softO\left( \left(d^{5n} + d^{4n}l  + d^{3n}l^2 \left(\tfrac{l}{n}\right)^{2n}\right)\eexp^{3n}\delta\right)
  \end{equation*}
arithmetic  operations in~$k$, where $\delta$ is the larger of~$\delta(a)$ and~$\delta(f)$, uniformly in all the parameters. Asymptotically with~$l$ and~$n$ fixed, this is~$\softO\left(d^{5n}\delta\right)$. 
\end{thm}
Note that while this may seem a huge complexity, it is not so bad when compared to the size of the output, which seems to be, empirically, comparable to~$d^{3n}\delta$, with~$n$ fixed and~$l=1$. 
Note also that for~$n=1$, the complexity improves over that of the algorithm based on Hermite's reduction studied in~\cite{BosCheChy10}, thanks to our avoiding too many rank computations.

\begin{remark}\label{rem:size-cert}\normalfont
Let $a/f$ be a generic fraction with a telescoper~$T$ and a regular certificate $A$.
We claim that the size of~$A$ is asymptotically bounded below by~$d^{(1-o(1))n^2} \delta$, making it crucial to avoid the computation of certificates.
Indeed, the fraction~$T(a/f)$ writes~$b/f^{r+1}$, where~$r$ is the order of~$T$.
The number of monomials of~$b$ in~$\xx$ is~$\binom{(r+1)d-1}{n}\approx(rd)^n/n!$.
  If~$a$ is generic then~$r$ is at least~$\dim\Hfp$, by the Cyclic Vector Theorem;
  and if~$f$ is generic, it satisfies~\eqref{eqn:regularity} and~$\dim\Hfp$ is about~$d^n$, by Corollary~\ref{coro:hdim}.
  Since~$T(a/f)$ equals~$\sum_i\partial_i(A_i)$, the size of the~$A_i$ has at least the same order of magnitude than that of~$T(a/f)$; hence the claim.
\end{remark}

\subsection{Primitives of linear algebra}\label{sec:pol-linalg}

\noindent The complexity of Algorithm~\ref{algo:ct} lies in operations on matrices with polynomial coefficients.
Let~$A\in k[t]^{n\times m}$ have rank~$r$ and coefficients of degree at most~$d$.
One can compute~$r$, a basis of~$\ker A$ and a maximal rank minor in~$\softO(n m r^{\omega-2} d)$ operations in~$k$~\cite{Zho13}.
A maximal rank minor can be inverted in complexity~$\softO(r^3 d)$~\cite{JeaVil05}.
In particular, a matrix~$B$ such that~$ABA=A$ can be computed in~$ \softO(n m r^{\omega-2} d + r^3 d)$ operations in~$k$, or~$\softO(n^2md)$, using~$r\leqslant n,m$ and~$\omega\leqslant3$.

A matrix~$A$ with rational entries is represented with the l.c.m.~$g$ of the entries and the polynomial matrix~$g A$. 
\subsection{Precomputation}

\noindent Algorithm~\ref{algo:ct} needs the splits~$\psi_{id}$ for~$i$ from~$1$ to the larger of~$n+1$ and~$l$.
Following~\S\ref{sec:red-mod-jac}, it is enough to compute~$\psi_{id}$ for~$i$ between~$1$ and~$n+1$, each for a cost of~$\softO(R_{id}C_{id}^2\delta(f))$ operations in~$k$,
and then~$\psi_{id}$ can be obtained with no further arithmetic operation for~$i>n+1$.
Thus the precomputation needs~$\softO\left( \eexp^{3n}d^{3n}\delta(f) \right)$  operations in~$k$.

\subsection{Reduction}

\noindent Let~$\rho(\ell, \delta(a))$ be the complexity of the variant of the algorithm $\Reduce$ based on linear algebra with input a rational function~$a/f^\ell$.
The procedure first computes~$\psi_{\ell d}(a)$. Since~$\psi_{\ell d}$ is precomputed, it is only the product of a matrix of dimensions~$C_{\ell d}$ by~$R_{\ell d}$ with the vector of coefficients of~$a$ in a monomial basis.
The elements of the matrix have degree at most~$\delta_E$ and the elements of the vector have degree at most~$\delta(a)$.
Thus the product has complexity~$\softO(R_{\ell d}C_{\ell d}(\delta(a)+\delta_E))$.
Secondly, the procedure computes~$r$ as~$\pi_{\ell d}(a)$ knowing~$\psi_{\ell d}(a)$; this has the same complexity.
Thirdly, it computes~$F_1$, computation whose complexity is dominated by that of the first step. And lastly it computes~$\Reduce(F_1)$, which has complexity bounded by~$\rho(\ell-1, \delta(a)+\delta_E)$.
Unrolling the recurrence leads to
\[ \rho(\ell, \delta(a)) = \softO\left( \ell \left(  \frac{\eexp d \ell}{n+1} \right)^{2n} \left( \delta(a) + \ell \delta_E \right) \right). \]

\subsection{Main loop}

\noindent The computation of~$G_0$ has complexity~$\rho(\ell, \delta(a))$.
Next, $G_i$ has shape given by~\eqref{bound:deg-diff}, and is differentiated before being reduced, so that the cost of 
the computation of~$G_{i+1}$ is at most~$\rho(n+1, \delta(a)+(i+2n)\delta(f)+(i+l)\delta_E)$.
Summing up, the computation of~$G_0,\dotsc,G_r$ has a complexity
\begin{equation}\label{eq:complexity_all_reductions}
 \rho(l,\delta(a))+ \softO\left( (\eexp d)^{2n}r \left (\delta(a)+r\delta(f)+(r+l)\delta_E \right)  \right). \end{equation}

During the~$i$th step, the procedure computes the rank of~$i+1$ vectors with~$\bigO(\eexp^n d^n)$ coefficients of degree~$\delta(b'_i)$ and computes a linear dependence relation if there is one.
This is done in complexity~$\softO\left( i^{\omega-1}\eexp^nd^n \delta(b'_i) \right)$.
This step is quite expensive and doing it for all~$i$ up to~$r$ would ruin the complexity.
It is sufficient to perform this computation only when~$i$ is a  power of~2 so that the maximal~$i$ which is used is smaller than~$2r$.
When the rank of the family is not full, we deduce from it the exact order~$r$ and perform the computation in that order.
Indeed, the rank over~$L$ of~$G_0,\dotsc,G_i$ is the least of~$r$ and~$i$.
This way, finding the rank and solving has cost~$\softO(r^{\omega-1}\eexp^nd^n(\delta(b'_r)+\delta(b'_{2r})))$. In view of~\eqref{eqn:deltabk'} and since~$r\leqslant d^n$ and~$\omega\leqslant 3$, the complexity of that step is bounded by~\eqref{eq:complexity_all_reductions}.
Adding the cost of the precomputation and using the bounds of the previous section leads to Theorem~\ref{thm:complexity-reg}.

\section{Affine singular case}
\label{sec:singular}

\noindent Let~$L$ denote the field~$k(t)$.
Let~$F_\aff$ be a rational function in~$L(x_1,\dotsc,x_n)$, written as~$a/f_\aff$.
We do not assume that~$F_\aff$ is homogeneous, nor that~$f_\aff$ satisfies a regularity property.
Let~$d_\aff$ be the total degree of~$f_\aff$ w.r.t.~$\xx$.

In this section we show a deformation technique that regularizes singular cases.
In particular, it allows to transfer the previous results to the general case and obtain the following bounds.
The algorithm is again based on linear algebra.

\begin{thm}\label{thm:bound-affine}
  The function~$F_\aff$ admits a telescoper, with regular certificate, of order at most~$d^n$ and degree
  \[ \bigO\left( d_\pr^n \delta(a) + d_\pr^{3n}\eexp^n \delta(f_\aff) \right), \]
  where~$d_\pr$ is~$\max(d_\aff,\deg_{\xx}a + n + 1)$.
  This telescoper can be computed in complexity~$\softO\left(\eexp^{3n}d_\pr^{8n}\delta \right)$, with~$\delta$ the larger of~$\delta(a)$ and~$\delta(f_\aff)$.
\end{thm}
 It is easy to see that the \emph{bit} complexity is also polynomial in~$d_\pr^n$.
The dependence in~$n$ of the complexity, with~$\deg_\xx a$ and~$d_\aff$ fixed, can be improved to~$\eexp^{O(n)}$ rather than~$n^{O(n)}$ with a more careful analysis.

\subsection{Homogenization and deformation}

\noindent The regularization proceeds in two steps.
First, let~$F_\pr$ be the homogenization of~$F_\aff$ in degree~$-n-1$, that is
\[ F_\pr = {x_0^{-n-1}}F_\aff\left( \tfrac{x_1}{x_0},\dotsc,\tfrac{x_n}{x_0} \right), \]
which we write~$b/f_\pr$ for some homogeneous polynomials~$b$ and~$f_\pr$. Let~$d_\pr$ denote the degree of~$f_\pr$;
it is given by Theorem~\ref{thm:bound-affine}.
The degrees of~$b$ and $f_\pr$ satisfy the hypothesis of Theorem~\ref{coro:asy-deg}, by construction, but in general~$f_\pr$ does not satisfy Hypothesis~\eqref{eqn:regularity}.
(Although it does generically, as long as~$d_\pr$ equals~$d_\aff$.)
We consider a new indeterminate~$\varepsilon$, the polynomial~$f_\reg$ defined by
\[ f_\reg = f_\pr + \varepsilon\sum_{i=0}^{\smash{n}} x_i^{d_\pr}, \]
and the rational function~$F_\reg$ defined by~$b/f_\reg$.
We could also have perturbed the square-free part of~$f_\pr$ rather than~$f_\pr$, leading to an improvement of the complexity in Theorem~\ref{thm:bound-affine} at the cost of more technical details.

\begin{lem}
The polynomial~$f_\reg$ satisfies Hypothesis~\eqref{eqn:regularity} over~$L(\varepsilon)$, that is~$\Lwx/\!\Jac f_\reg$ has finite dimension.
\end{lem}
\begin{proof}
This is true for~$\varepsilon=\infty$, so it is generically true. 
\end{proof}
Now, Theorem~\ref{coro:asy-deg} gives bounds on the order and degree of a telescoper of~$F_\reg$, which is in~$L(\varepsilon)[\xx,\tfrac{1}{f_\reg}]_{\smash{-n-1}}$. The proof of Theorem~\ref{thm:bound-affine} is concluded by the following.
\begin{prop}\label{prop:telesc-regtosing}
  If~$T$ in~$L[\varepsilon]\langle \partial_t \rangle$ is a telescoper of~$F_\reg$ with regular certificate, then so is~$T|_{\varepsilon=0}$ for~$F_\aff$.
\end{prop}
\begin{proof}
By assumption, $T(F_\reg)$ equals~$\sum_{i=0}^n\partial_i g_i/f_\reg^p$ for some integer~$p$ and polynomials~$g_i$ in~$\Lwx$.
Each~$g_i/f_\reg^p$ can be expanded in Laurent series in~$\varepsilon$ as~$\sum_{j\geqslant N} h_{ij}\varepsilon^j$ for some possibly negative integer~$N$ and rational functions~$h_{ij}$ in~$L[\xx,\tfrac{1}{f_\pr}]_{-n}$.
Similarly, we can write the operator~$T(F_\reg)$ as~$T|_{\varepsilon=0}(F_\pr)+\varepsilon\sum_{j\geqslant 0} b_j \varepsilon^j$ for some rational functions~$b_j$ in~$L[\xx,\tfrac{1}{f_\pr}]$.
Since the derivations~$\partial_i$ commute with~$\varepsilon$, it is clear that~$T|_{\varepsilon=0}(F_\pr)$ equals~$\sum_{i=0}^n \partial_i h_{i0}$.
Next, in this equality, $x_0$ can be evaluated to~1 to give
\[ T|_{\varepsilon=0}(F_\aff) = (\partial_0 h_{00})|_{x_0=1}+ \sum_{i=1}^{\smash{n}} \partial_i (h_{i0}|_{x_0=1}). \]
Euler's relation for~$h_{00}$ gives (with the index~$00$ dropped)
\begin{align*}
 (\partial_0h)|_{x_0=1}  &= - \sum_{i=1}^{\smash{n}} \partial_i( x_i h|_{x_0=1}),
\end{align*}
proving that~$(\partial_0h)|_{x_0=1}$ is in~$D_{f_\aff}$. Thus, so is $T|_{\varepsilon=0}(F_\aff)$ and the proof is complete.
\end{proof}
Nevertheless, a telescoper obtained in this way does not need to be minimal, even starting from a minimal one for the perturbed function~$F_\reg$.
This is unfortunate because in presence of singularities the dimension of~$\Hfp$ can collapse when compared to the generic order given by Corollary~\ref{coro:hdim}.

\subsection{Algorithm and complexity}

\noindent The algorithm is based on Proposition~\ref{prop:telesc-regtosing}.
We use an evaluation-interpolation scheme to control the complexity.
Let the operator~$T$ in~$k(t,\varepsilon)\langle \partial_t \rangle$ be the minimal telescoper of~$F_\reg$, written as~$\partial_t^r + \sum_{k=0}^{r-1} \frac{c_k}{c_r}\partial_t^k$. It is the output of~$\CreatTel$ applied to~$F_\reg$. We aim at computing~$(\varepsilon^\alpha T)|_{\varepsilon=0}$, where~$\alpha$ is such that this evaluation is finite and not zero.

Proposition~\ref{prop:telesc-regtosing}, slightly adapted, shows that~$T|_{\varepsilon=u}$ is a telescoper with regular certificate of~$F_\reg|_{\varepsilon=u}$ whenever~$c_r(t,u)$ is not zero, even if~$f_\reg|_{\varepsilon=u}$ does not satisfy~\eqref{eqn:regularity}.
When it does, the specialization gives the minimal one:
\begin{lem}\label{lem:telesc-spe}
If~$f_\reg|_{\varepsilon=u}$ satisfies hypothesis~\eqref{eqn:regularity}
and if~$u$ does not cancel~$c_r$,
then~$T|_{\varepsilon=u}$ is the minimal telescoper with regular certificate of~$F_\reg|_{\varepsilon=u}$.
\end{lem}

\begin{proof}
  We use  the notation of Section~\ref{sec:eff-bounds}, replacing~$f$ by~$f_\reg$ and~$L$ by~$L(\varepsilon)$.
  The operator~$T$ is the output of Algorithm~\ref{algo:ct} applied to~$F_\reg$.
  Since~$f_\reg|_{\varepsilon=u}$ satisfies~\eqref{eqn:regularity}, for all~$d$ the matrix~$\Mat\varphi_d$, with coefficients in~$L[\varepsilon]$, has the same rank as its specialization with~$\varepsilon=u$~\cite[\S 58]{Mac16}.
  Thus, to compute the splits~$\psi_d$ we can choose maximal rank minors of~$\Mat\varphi_d$ that are also maximal rank minors of the specialization.
  When doing so, the reduction~$[\cdot]$ commutes with the evaluation~$\cdot|_{\varepsilon=u}$.
  In particular, the polynomials~$E_q$ do not vanish for~$\varepsilon=u$.

In the proof of Prop.~\ref{thm:delta-ck}, Eq.~\eqref{eqn:cr-div-det} shows that~$c_r$, the leading coefficient of~$T$, divides~$P_{l+r} \det B$. The polynomial~$P_{l+r}$ is a product of several~$E_d$'s, in particular~$P_{l+r}|_{\varepsilon=u}$ is not zero.
  Since~$c_r|_{\varepsilon=u} \neq 0$, the determinant of~$B|_{\varepsilon=u}$ is not zero either.
  Looking at the definition of~$B$ in the proof of Prop.~\ref{thm:delta-ck}, this implies that the~$[\partial^i_t F_\reg]|_{\varepsilon=u}$, for~$i$ between~$0$ and~$r-1$ are free over~$L(\varepsilon)$.
  In particular, a telescoper with regular certificate of~$F_\reg|_{\varepsilon=u}$ has order at least~$r$.
  Since~$T|_{\varepsilon=u}$ is a telescoper of order is~$r$, it is the minimal one.
\end{proof}

We now present the algorithm.
Let~$N$ be~$\eexp^n(d^{3n}+d^{2n}+d^n)$.
By Proposition~\ref{thm:delta-ck}, the polynomials~$c_k$ have degree at most~$N$ in~$\varepsilon$, and at most~$N \delta$ in~$t$.
Choose a set~$U$ of~$4N+1$ elements of~$k$.
Determine the set~$U'$ of elements~$u$ of~$U$ such that~$f_\reg|_{\varepsilon=u}$ satisfies~\eqref{eqn:regularity}.
This step has complexity~$\softO\left( (\eexp d)^{n\omega}\delta|U| \right)$:
The polynomial~$f_\reg|_{\varepsilon=u}$ satisfies~\eqref{eqn:regularity} if and only if~$(\Mat\varphi_{D})|_{\varepsilon=u}$ is full rank.
In particular, if~$f_\reg|_{\varepsilon=u}$ does not satisfy~\eqref{eqn:regularity}, then~$E_{D}|_{\varepsilon=u}$ vanishes.
The polynomial~$E_{D}$ has degree at most~$\eexp^n d^n$ in~$\varepsilon$, by Proposition~\ref{prop:deltaE}, so~$U\setminus U'$ has at most~$\eexp^n d^n$ elements.
For each~$u$ in~$U'$, compute~$\CreatTel(f_\reg|_{\varepsilon=u})$ with leading coefficient normalized to~$1$, denoted by~$T_u$.
This step has complexity~$\softO(d^{5n}\eexp^{3n}\delta |U'|)$, by Theorem~\ref{thm:complexity-reg}.
Determine the subset~$U''$ of~$U'$ where the order of~$T_u$ is maximal.
By Lemma~\ref{lem:telesc-spe}, the complement~$U'\setminus U''$ is formed by~$u$ such that~$c_r(t,u)=0$.
It has at most~$N$ elements since~$c_r$ has degree at most~$N$ in~$\varepsilon$.
For all~$u$ in~$U''$ the operators~$T_u$ and~$T|_{\varepsilon=u}$ coincide.
Thus~$U''$ has at most~$2N + 1$ elements.

The~$r$ rational functions~$\frac{c_k(t,0)}{c_r(t,0)}$ can be computed using Lemma~\ref{prop:rat-reval} in total complexity~$\softO( N^2r\delta )$.
If~$c_r(t,0)$ is zero, we look for the positive integer~$\alpha$ such that the functions~$\varepsilon^\alpha\frac{c_k(t,\varepsilon)}{c_r(t,\varepsilon)}$ are finite for~$\varepsilon=0$ but not zero for at least one~$k$.
The integer~$\alpha$ is at most~$N$ and thus can be found with a binary search, using at most~$\log_2 N+1$ times Lemma~\ref{prop:rat-reval}.

\begin{lem}\label{prop:rat-reval}
  Let~$R$ in~$k(x,y)$ be written~$P/Q$, with~$P$ and~$Q$ polynomials of degree less than~$d_x$ in~$x$ and~$d_y$ in~$y$.
  Given evaluations~$R(x,v)$, for~$2 d_y+1$ elements $v$ of~$k$, the function~$R(x,0)$ (or~$\infty$ if~$Q(x,0)$ vanishes) can be computed using~$\softO(d_x d_y)$ arithmetic operations in~$k$.
\end{lem}

\begin{proof}
  Let~$V$ the set of evaluation points.
  Choose a set~$U$ of~$2d_x+1$ points of~$k$.
  Compute~$R(u,v)$ for~$u\in U$ and~$v \in V$ in~$\softO(d_x d_y)$ operations.
  Note that there is no need to check that the elements of~$U$ are not poles of the~$R(x,v)$: univariate rational reconstruction can handle that.
  Use univariate rational reconstruction to compute~$R(u,y)$, for~$u$ in~$U$, in complexity~$\softO(d_y |U|)$ operations.
  Reconstruct~$R(x,0)$ in complexity~$\softO(d_x)$ from the evaluations~$R(u,0)$.
\end{proof}

\section{Experiments}

\noindent A basic implementation of the algorithm \tCreatTel\ has been written in Maple 16.
As it uses only Maple primitives to compute with polynomial matrices, it is certainly too basic to reflect the complexity given in Theorem~\ref{thm:complexity-reg}.

Table~\ref{tab:orddeg-exp} presents empirical results for some generic rational functions, with~$n=2$. 
The bound on the order are generically exact as expected;
however the bound on the degree is not very sharp.
For~$n=1$ and~$\delta(a)$ fixed, a careful study~\cite{BosCheChy10} proves that the degree of the minimal telescoper is~$\bigO(d^2\delta)$, which is tighter than the~$\bigO(d^3\delta)$ given by Theorem~\ref{coro:asy-deg}.
Analogy, as well as numerical evidence and theoretical clues, lead us to think that for general~$n$, the asymptotic behavior can be improved from~$\bigO(d^{3n}\delta)$ to~$\bigO(d^{2n}\delta)$.

The relative cost of each step of Algorithm~\ref{algo:ct} in the computation of telescopers of Table~\ref{tab:orddeg-exp}, on the example of the telescoper of degree~12 and degree~1092 of a generic function~$a/f^2$ as described in Table~\ref{tab:orddeg-exp}, that is computed in about~7 hours breaks down as follows:
The computation of splits of Macaulay matrices takes about~1\% of the time, the reduction steps about~40\%, and the final solving  about~60\% of the time.
More efficient matrix multiplication and system resolution over univariate polynomials could improve speed dramatically.
We have not been able to compute more than the first column of Table~\ref{tab:orddeg-exp} with methods and programs in~\cite{Kou10,CheKauSin12}.

On the other hand, the regularity hypothesis~\eqref{eqn:regularity} is restrictive in applications:
Even though generic polynomials satisfy this hypothesis, examples with physical or combinatorial meaning usually do not.
The method shown in Section~\ref{sec:singular} is only of theoretical interest.
By contrast, the algorithm for the regular case is very efficient in practice.%

\scriptsize
\bibliographystyle{abbrv}
\bibliography{BoLaSa-issac}

\begin{thebibliography}{10}

\bibitem{AbbKedRoe10}
T.~G. Abbott, K.~S. Kedlaya, and D.~Roe.
\newblock Bounding {P}icard numbers of surfaces using {$p$}-adic cohomology.
\newblock In {\em {AGCT} 2005}, volume~21 of {\em S\'emin. Congr.}, pages
  125--159. SMF, Paris, 2010.

\bibitem{ApaZei06}
M.~Apagodu and D.~Zeilberger.
\newblock Multi-variable {Z}eilberger and {A}lmkvist-{Z}eilberger algorithms
  and the sharpening of {W}ilf- {Z}eilberger theory.
\newblock {\em Adv. in Appl. Math.}, 37(2):139--152, 2006.

\bibitem{BosCheChy10}
A.~Bostan, S.~Chen, F.~Chyzak, and Z.~Li.
\newblock Complexity of creative telescoping for bivariate rational functions.
\newblock In {\em I{SSAC}'10}, pages 203--210. ACM, 2010.

\bibitem{CheKauSin12}
S.~Chen, M.~Kauers, and M.~F. Singer.
\newblock Telescopers for rational and algebraic functions via residues.
\newblock In {\em I{SSAC}'12}, pages 130--137. ACM, 2012.

\bibitem{Chr83}
G.~Christol.
\newblock Diagonales de fractions rationnelles et equations diff\'erentielles.
\newblock In {\em Groupe de travail d'analyse ultramétrique, 1982/83}, volume
  12, {issue\ 2, exp.\ 18}, pages 1--10. Paris, 1984.

\bibitem{Chr85}
G.~Christol.
\newblock Diagonales de fractions rationnelles et \'equations de
  {P}icard-{F}uchs.
\newblock In {\em Groupe de travail d'analyse ultramétrique, 1984/85}, volume
  12, {issue\ 1, exp.\ 13}, pages 1--12. Paris, 1985.

\bibitem{Chy00}
F.~Chyzak.
\newblock An extension of {Z}eilberger's fast algorithm to gen- er\-al
  holonomic functions.
\newblock {\em Disc. Math.}, 217(1-3):115--134, 2000.

\bibitem{Dwo62}
B.~Dwork.
\newblock On the zeta function of a hypersurface.
\newblock {\em IHES Publ. Math.}, (12):5--68, 1962.

\bibitem{Dwo64}
B.~Dwork.
\newblock On the zeta function of a hypersurface. {II}.
\newblock {\em Ann. of Math. (2)}, 80:227--299, 1964.

\bibitem{GatGer03}
J.~\gathen{von zur} Gathen and J.~Gerhard.
\newblock {\em Modern Computer Algebra}.
\newblock Cambridge University Press, Cambridge, second edition, 2003.

\bibitem{Gerkmann07}
R.~Gerkmann.
\newblock Relative rigid cohomology and deformation of hypersurfaces.
\newblock {\em Int. Math. Res. Pap.}, (1):Art. 3, 1--67, 2007.

\bibitem{Gri69}
P.~A. Griffiths.
\newblock On the periods of certain rational integrals. {I}, {II}.
\newblock {\em Ann. of Math}, 90(3):460--495, 496--541, 1969.

\bibitem{JeaVil05}
C.-P. Jeannerod and G.~Villard.
\newblock Essentially optimal computation of the inverse of generic polynomial
  matrices.
\newblock {\em J. Complexity}, 21(1):72--86, 2005.

\bibitem{KauZei11}
M.~Kauers and D.~Zeilberger.
\newblock The computational challenge of enumerating high-dimensional rook
  walks.
\newblock {\em Adv. in Appl. Math.}, 47(4):813--819, 2011.

\bibitem{Kou10}
C.~Koutschan.
\newblock A fast approach to creative telescoping.
\newblock {\em Math. Comput. Sci.}, 4(2-3):259--266, 2010.

\bibitem{Laz77}
D.~Lazard.
\newblock Alg{\`e}bre lin{\'e}aire sur {$K[X_{1},\cdots,X_{n}]$}, et
  {\'e}limination.
\newblock {\em Bull. Soc. Math. France}, 105(2):165--190, 1977.

\bibitem{Lip88}
L.~Lipshitz.
\newblock The diagonal of a {$D$}-finite power series is {$D$}-finite.
\newblock {\em J. Algebra}, 113(2):373--378, 1988.

\bibitem{Mac16}
F.~S. Macaulay.
\newblock {\em The algebraic theory of modular systems}, volume XXXI of {\em
  Cambridge Mathematical Library}.
\newblock Cambridge University Press, 1994.
\newblock Revised reprint of the 1916 original.

\bibitem{Mon70}
P.~Monsky.
\newblock {\em {$p$}-adic analysis and zeta functions}, volume~4 of {\em
  Lectures in Mathematics, Department of Mathematics, Kyoto University}.
\newblock Kinokuniya Book-Store Co. Ltd., Tokyo, 1970.

\bibitem{Mor92}
D.~R. Morrison.
\newblock Picard-{F}uchs equations and mirror maps for hypersurfaces.
\newblock In {\em Essays on mirror manifolds}, pages 241--264. Int. Press, Hong
  Kong, 1992.

\bibitem{Mov11a}
H.~Movasati.
\newblock {\em Multiple integrals and modular differential equations}.
\newblock IMPA Mathematical Publications. Instituto Nacional de Matem\'atica
  Pura e Aplicada (IMPA), Rio de Janeiro, 2011.

\bibitem{Pan10}
S.~Pancratz.
\newblock Computing {G}auss--{M}anin connections for families of projectives
  hypersurfaces.
\newblock A thesis submitted for the Transfer of Status, Michaelmas, 2009.

\bibitem{Pic99}
{\'E}.~Picard.
\newblock Quelques remarques sur les int\'egrales doubles de seconde esp\`ece
  dans la th\'eorie des surfaces alg\'ebriques.
\newblock {\em C. R. Acad. Sci. Paris}, 129:539--540, 1899.

\bibitem{Pic02}
{\'E}.~Picard.
\newblock Sur les p\'eriodes des int\'egrales doubles et sur une classe
  d'\'equations diff\'erentielles lin\'eaires.
\newblock {\em C. R. Acad. Sci. Paris}, 134:69--71, 1902.

\bibitem{PicSim}
{\'E}.~Picard and G.~Simart.
\newblock {\em Th{\'e}orie des fonctions alg{\'e}briques de deux variables
  ind{\'e}pendantes}.
\newblock {G}authier-{V}illars, 1906.
\newblock Tome II.

\bibitem{Takayama1990a}
N.~Takayama.
\newblock An algorithm of constructing the integral of a module --- an infinite
  dimensional analog of {G}r{\"o}bner basis.
\newblock In {\em ISSAC'90}, pages 206--211. ACM, 1990.

\bibitem{WilZei92}
H.~S. Wilf and D.~Zeilberger.
\newblock An algorithmic proof theory for hypergeometric (ordinary and
  ``{$q$}'') multisum/integral identities.
\newblock {\em Invent. Math.}, 108(3):575--633, 1992.

\bibitem{Zei90}
D.~Zeilberger.
\newblock A holonomic systems approach to special func- tions identities.
\newblock {\em J. Comput. Appl. Math.}, 32(3):321--368, 1990.

\bibitem{Zei91}
D.~Zeilberger.
\newblock The method of creative telescoping.
\newblock {\em J. Symb. Comp.}, 11(3):195--204, 1991.

\bibitem{Zho13}
W.~Zhou.
\newblock {\em Fast Order Basis and Kernel Basis Computation and Related
  Problems}.
\newblock PhD thesis, Univ. of Waterloo, 2013.

\end{thebibliography}

\end{document}